\documentclass[11pt]{article}

\usepackage[dvipsnames]{xcolor}
\usepackage[utf8]{inputenc}
\usepackage[T1]{fontenc}
\usepackage{amsmath}
\usepackage{amsfonts}
\usepackage{ae}
\usepackage{units}
\usepackage{icomma}
\usepackage{color}
\usepackage{graphicx}
\usepackage{bbm}
\usepackage{amsthm}
\usepackage{amssymb}
\usepackage{cancel}
\usepackage{upgreek}
\usepackage{eso-pic}
\usepackage[breaklinks,pdfpagelabels=false]{hyperref}
\usepackage{xspace}
\usepackage{amsrefs}
\usepackage{tikz}
\usepackage{thmtools}
\usepackage{thm-restate}
\usepackage{cleveref}
\usepackage{enumerate}
\usepackage{dsfont}
\usepackage{verbatim}
\usepackage{parskip}
\usepackage{float}
\usepackage[letterpaper,margin=1.1in]{geometry}

\usetikzlibrary{patterns}
\tikzset{
  schraffiert/.style={pattern=horizontal lines,pattern color=#1},
  schraffiert/.default=black
}

\newcommand{\eps}{\epsilon}
\newcommand{\dist}{\operatorname{dist}}

\newtheorem{theorem}{Theorem}[section]

\newtheorem{claim}[theorem]{Claim}
\theoremstyle{definition}

\definecolor{orange}{rgb}{1,0.5,0}

\usepackage{faktor}

\numberwithin{equation}{section}
\numberwithin{theorem}{section}

\begin{document}
\setlength{\parskip}{8pt plus 0.5ex minus 0.2ex}
\renewcommand{\baselinestretch}{1.1}

\title{Optimal strategies for patrolling fences\footnote{This work started and was majorly produced thanks to Angelika Steger's 2018 ETHZ Buchboden Workshop on Combinatorial Structures and Algorithms}}

\author{Bernhard Haeupler\footnote{\texttt{haeupler@cs.cmu.edu}; Part of this work was done in the summer of 2018 as a visiting professor at ETH Z\"urich supported by Mohsen Ghaffari. Supported also in part by NSF grants CCF-1618280, CCF-1814603, CCF-1527110 and NSF CAREER award CCF-1750808.}\\Carnegie Mellon University\\
\and
Fabian Kuhn\footnote{\texttt{kuhn@cs.uni-freiburg.de}; Supported in part by ERC Grant No. 336495 (ACDC).}\\University of Freiburg\\
\and
Anders Martinsson\footnote{\texttt{maanders@inf.ethz.ch}}
, Kalina Petrova\footnote{\texttt{kpetrova@student.ethz.ch}; Supported in part by Excellence Scholarship \& Opportunity Programme of ETH Z\"urich.}, and Pascal Pfister\footnote{\texttt{ppfister@inf.ethz.ch}}\\ETH Z\"urich
}

\date{}

\maketitle

\begin{abstract}
A classical multi-agent fence patrolling problem asks: What is the maximum length $L$ of a line fence that $k$ agents with maximum speeds $v_1,\ldots, v_k$ can patrol if each point on the line needs to be visited at least once every unit of time.
It is easy to see that $L = \alpha \sum_{i=1}^k v_i$ for some efficiency $\alpha \in [\frac{1}{2},1)$. After a series of works~\cite{czyzowicz2011boundary,dumitrescu2014fence,kawamura2015fence,kawamura2015simple} giving better and better efficiencies, it was conjectured by Kawamura and Soejima \cite{kawamura2015simple} that the best possible efficiency approaches $\frac{2}{3}$. No upper bounds on the efficiency below $1$ were known.

We prove the first such upper bounds and tightly bound the optimal efficiency in terms of the minimum speed ratio $s = \frac{v_{\max}}{v_{\min}}$ and the number of agents $k$. Our bounds of $\alpha \leq \frac{1}{1 + \frac{1}{s}}$ and $\alpha \leq 1 - \frac{1}{\sqrt{k}+1}$ imply that in order to achieve efficiency $1 - \eps$, at least $k \geq \Omega(\eps^{-2})$ agents with a speed ratio of $s \geq \Omega(\eps^{-1})$ are necessary. 
Guided by our upper bounds, we construct a scheme whose efficiency approaches $1$, disproving the conjecture stated above. Our scheme asymptotically matches our upper bounds in terms of the maximal speed difference and the number of agents used.

A variation of the fence patrolling problem considers a circular fence instead and asks for its circumference to be maximized. We consider the unidirectional case of this variation, where all agents are only allowed to move in one direction, say clockwise. At first, a strategy yielding $L = \max_{r \in [k]} r \cdot v_r$ where $v_1 \geq v_2 \geq \dots \geq v_k$ was conjectured to be optimal by Czyzowicz et al.~\cite{czyzowicz2011boundary} This was proven not to be the case by giving constructions for only specific numbers of agents with marginal improvements of $L$. We give a general construction that yields $L = \frac{1}{33 \log_e\log_2(k)} \sum_{i=1}^k v_i$ for any set of agents, which in particular for the case $1, 1/2, \dots, 1/k$ diverges as $k \rightarrow \infty$, thus resolving a conjecture by Kawamura and Soejima~\cite{kawamura2015simple} affirmatively.
\end{abstract}

\section{Introduction}

Patrolling is a fundamental task in robotics, multi-agent systems, and security settings. Given some environment of interest, and a collection of mobile agents, the aim is to coordinate the movements of the agents in order to, for example, guard an area from intrusion by an enemy, prevent accidents or failure of equipment, maintain up-to-date information of the environment, etc. For each of these tasks, ensuring that certain points in the environment get visited/monitored frequently is crucial. Performance of patrolling algorithms is consequently often measured in terms of \emph{idleness} -- roughly speaking, the time between two consecutive visits to a point in the environment.

Multi-agent patrolling has been extensively studied in the robotics literature since the early 2000s, e.g., see \cite{chevaleyre2004theoretical,machado2003empirical} and the survey \cite{portugal2011survey}. However, even for extremely clean and very simple models, determining optimal patrolling schemes poses many natural mathematical questions with interesting and surprisingly sophisticated answers \cite{collins2013, czyzowicz2011boundary,  czyzowicz2017faulty,  czyzowicz2017twospeed, czyzowicz2017trees, czyzowicz2014visibility,  dumitrescu2014fence, kawamura2015fence, kawamura2015simple}. 

\subsection{Fence Patrolling} 
This paper studies a classical fence patrolling problem introduced by Czyzowicz et al.~\cite{czyzowicz2011boundary}, which might be one of the cleanest and most natural patrolling problems: What is the maximum length $L$ of a fence that $k$ agents $a_1, \dots, a_k$ with maximum speeds $v_1, \ldots, v_k$ can patrol if each point needs to be visited at least once every unit of time. 
Czyzowicz et al. introduce two variations of this question -- the fence could be either an open curve, or a closed curve. For simplicity, we assume the open curve is a line segment and the closed curve is a circle.

For the line segment, it is easy to see that for any speeds the maximum length $L$ satisfies $L = \alpha \sum_{i=1}^k v_i$ for some efficiency $\alpha \in [\frac{1}{2},1)$. In particular, in one unit of time an agent $a_i$ can cover a length of at most $v_i$ and all agents can cover at most a total length of $\sum_{i=1}^k v_i$.  An efficiency of exactly $\alpha = 1$ is furthermore never possible because agents have to turn around eventually. On the other hand, an efficiency of $\alpha=\frac{1}{2}$ can easily be achieved by the following strategy:\\~\\
\textsc{Partition-based strategy, $\mathcal{A}_1$:} For all $i \in [k],$ agent $a_i$  patrols a subsegment of length $\frac{1}{2} v_i$ by going back and forth on this segment once every unit of time. This patrols a segment of length $L=\frac{1}{2}\sum_{i=1}^k v_i$ with idle time $1$.

Considering patrol schedules on a circle, the picture is quite different than for a line segment. Again, the length $L$ of any circle that can be patrolled by a set of agents is upper-bounded by the sum of the maximum speeds of the agents, since agent $a_i$ cannot cover a length of more than $v_i$. Here, however, it is easy to find collections of agents and a corresponding patrol schedule that achieves this exactly -- imagine $k$ identical agents starting equidistantly along the circle and moving in unison in the same direction, say counter-clockwise. 

\subsection{Prior Work on Fence Patrolling} 

\subsubsection{Prior Work on the Line Segment}

Czyzowicz et al.~\cite{czyzowicz2011boundary} observed that the trivial scheme $\mathcal{A}_1$ with efficiency $\frac{1}{2}$ is optimal if the paths of the agents never cross. To see this, note that the leftmost agent $a_i$ cannot walk away further than $\frac{1}{2} v_i$ from the leftmost point of the fence as it would take more than one unit of time between two visits of this point. By the same argument the agent $a_j$ to the right of agent $a_i$ cannot ever be further away than $\frac{1}{2} (v_i + v_j)$ from the leftmost point of the fence and induction shows that a total fence length of $\frac{1}{2}\sum_{i=1}^k v_i$ is best possible. For the special case of all agents having the same speed the assumption that the paths of the agents never cross is furthermore without loss of generality as one can equally well switch identities of agents at a crossing, making the agents bounce off each other instead of crossing. In the worst case an efficiency of $\frac{1}{2}$ is thus optimal and Czyzowicz et al. posited~\cite{czyzowicz2011boundary} that indeed no better efficiency can be achieved for any speeds. 

Surprisingly, Kawamura and Kobayashi~\cite{kawamura2015fence} disproved this by providing an explicit fence patrolling schedule for $6$ agents with speeds $1, 1, 1, 1, \frac{7}{3},$ and $\frac{1}{2}$ for a fence of length $\frac{7}{2}$, thus achieving an efficiency of $\frac{21}{41}>\frac{1}{2}$. This was improved by Dumitrescu, Ghosh and T\'oth \cite{dumitrescu2014fence}, who proposed a family of patrolling schedules with efficiency approaching $\frac{25}{48}$, and finally by Kawamura and Soejima \cite{kawamura2015simple} who achieved an efficiency approaching $\frac{2}{3}$. Kawamura and Soejima furthermore explicitly conjectured that no efficiency better than $\frac{2}{3}$ is possible for any set of speeds~\cite[Conjecture 6, page 9]{kawamura2015simple}. 

On the other hand, except for the setting of equal speeds discussed above, no upper bounds on the efficiency below $1$ have been provided in the literature~\cite{czyzowicz2011boundary,kawamura2015fence,dumitrescu2014fence,kawamura2015simple}.

\subsubsection{Prior work on the Circle}

For a general set of agents, Czyzowicz et al. \cite{czyzowicz2011boundary} proposed the following universal scheme that generalizes the aforementioned schedule for equal speeds:\\~\\
\textsc{Runners strategy, $\mathcal{A}_2$:} Assume $v_1\geq v_2 \geq \dots \geq v_k$. Find the $r\in[k] $ that maximizes $r\cdot v_r$, and let the $r$ fastest agents move equidistantly along the circle at speed $v_r$. This patrols a circle of length $L=\max_{r\in[k]} r\cdot v_r$ with idle time $1$.

Suppose for a collection of agents with maximum speeds $v_1\geq v_2 \geq \dots \geq v_k$, $\mathcal{A}_2$ produces a schedule on a circle with length $L$. Without loss of generality, we can assume $L=1$. Then $\max_{r\in[k]} r\cdot v_r = 1$, and by possibly increasing the maximum speed of some agents we may assume $v_i = 1/i$ for each $i\in[k]$. Note that increasing speeds in this way can only increase the maximum circumference that can be patrolled with idle time $1$ using these agents, but will not increase the length produced by $\mathcal{A}_2$. Thus, if there is any collection of agents where there is a patrol schedule that performs better than $\mathcal{A}_2$, there must be such a schedule in the case of \emph{harmonic} maximum speeds $1, 1/2, \dots, 1/k$.

To analyse the performance of patrol schemes on the circle, Czyzowicz et al. considered two different cases: \emph{unidirectional} patrol schedules, where agents are only allowed to move in one direction, and general (or \emph{bidirectional}) patrol schedules, where agents are allowed to go in both directions. Clearly, any patrol schedule obtained through $\mathcal{A}_2$ is unidirectional.

In the bidirectional case, it is not too hard to see that there are situations where $\mathcal{A}_2$ is not optimal. Indeed, in the case of harmonic maxiumum speeds, the partition-based strategy $\mathcal{A}_1$, which works in the same way for a circular fence as for a line segment fence, would give $L=(1+1/2+\dots+1/k)/2$, which is bigger than $1$ as given by $\mathcal{A}_2$ for any $k \geq 4$. In fact, an example with three agents was given in \cite{czyzowicz2011boundary} where neither $\mathcal{A}_1$ nor $\mathcal{A}_2$ are optimal. This was strengthened further by Dumitrescu et al. \cite{dumitrescu2014fence}, who showed for any $k\geq 4$ there exists a collection of $k$ agents where what they call the \emph{train strategy} $\mathcal{A}_3$ performs strictly better than both $\mathcal{A}_1$ and $\mathcal{A}_2$. 
To the authors' knowledge, no universal scheme has been proposed to always produce an optimal patrol schedule in this setting.

For the unidirectional case, it was initially conjectured by Czyzowicz et al. that $\mathcal{A}_2$ is optimal for any set of agents. This was proved to be true for up to four agents. However, it was shown incorrect by parallel results by Dumitrescu et al. \cite{dumitrescu2014fence} and Kawamura and Soejima \cite{kawamura2015simple}, who gave explicit examples of patrol schemes for $32$ and $122$ agents (with harmonic speeds) with $L=1+\epsilon$ (for a small unspecified $\epsilon>0$) and $L=1.05$ respectively. Kawamura and Soejima further conjectured that the maximum length of a unidirectional circle that can be patrolled by agents with speeds $1, 1/2, \dots 1/k$ diverges as $k\rightarrow\infty$.

\section{Our Results}

This paper advances the understanding of the fence patrolling problem by giving tight upper and lower bounds on the optimal efficiency for the line segment, and a construction for the circle with efficiency of $\Theta(\frac{1}{\log_e\log_2{k}})$ for any set of $k$ agents. To a large extent it concludes the main line of inquiry put forward in the works discussed above~\cite{czyzowicz2011boundary,kawamura2015fence,dumitrescu2014fence,kawamura2015simple}.

\subsection{Results for the Line Segment}

We provide the first technique to prove general impossibility results for the fence patrolling problem. We explain our ideas in more detail in Section~\ref{sec:impossiblity} and merely state our main upper bound here:

\begin{restatable}{theorem}{upperBoundFence}
\label{thm:GeneralUpperBoundFence}
Any fence patrol schedule with $k$ agents with maximum speeds $v_1, \dots, v_k$ patrols a fence of length at most 
$$L \leq \sum_{i=1}^{k} \frac{v_i}{1 + \frac{v_i}{\max_j v_j}}.$$
\end{restatable}

One way to interpret Theorem \ref{thm:GeneralUpperBoundFence} is that the contribution of an agent $a_i$ depends not only on his/her own speed $v_i$ but also on how much slower he/she is than the fastest agent. In particular, instead of always contributing $v_i$, as in the trivial upper bound, an agent contributes at most $\frac{1}{1 + \frac{1}{s_i}} \cdot v_i$ given that the fastest agent patrolling is a factor of $s_i$ faster than $a_i$. That is, the ``relative efficiency'' of an agent $a_i$ ranges anywhere between $1/2$ and $1$ depending on $s_i$, which always constitutes an improvement over the trivial upper bound of $\sum_i v_i$.

We also show that Theorem \ref{thm:GeneralUpperBoundFence} can be used to prove an upper bound on the efficiency of a schedule solely in terms of the number of agents:

\begin{restatable}{lemma}{secondCor}
\label{lem:UpperBoundNumber}
Any fence patrolling schedule with $k$ agents has an efficiency of at most $1 - \frac{1}{\sqrt{k} + 1}$.
\end{restatable}

We note that our upper bounds are tight in several interesting special cases. Specifically, for the case of agents having identical speeds, Theorem~\ref{thm:GeneralUpperBoundFence} shows that the efficiency of the schedule (and indeed each agent) is at most $\frac{1}{2}$, reproving the result of \cite{czyzowicz2011boundary}. In contrast to the symmetry argument about non-crossing agents explained above, our arguments and upper bounds easily extend to near-identical speeds as well. Lastly, it is easy to check that Theorem~\ref{thm:GeneralUpperBoundFence} is tight when applied to the configuration of agents used by Dumitrescu et al. \cite{dumitrescu2014fence} and  Kawamura and Soejima \cite{kawamura2015simple} for their construction to obtain efficiency ratios of $25/48-o(1)$ and $2/3-o(1)$, respectively.

Our upper bounds do not exclude schedules with efficiency close to $1$. They do however give important restrictions and clues about what an extremely efficient schedule, if it exists, has to look like. In particular, Lemma~\ref{lem:UpperBoundNumber} implies that any schedule with efficiency $1-\eps$ has to have at least $\left(\frac{1}{2\eps}\right)^2$, i.e., quadratically in $\frac{1}{\eps}$ many agents. In the same manner, Theorem~\ref{thm:GeneralUpperBoundFence} implies that, with $\eps \rightarrow 0$, the ratio between the fastest and slowest agent has to  be at least $\Omega(\frac{1}{\eps})$, i.e. grow unboundedly. Even more interestingly, the way the upper bound in Theorem~\ref{thm:GeneralUpperBoundFence} depends on $\max_i v_i$ seems to indicate that even just a single very fast agent can raise the ``relative efficiency'' of slower agents from $1/2$ to almost $1$.

Equipped with this better understanding and guidance from our impossibility results we were, to our surprise, able to design schedules which achieve an efficiency arbitrarily close to $1$, thus disproving the conjecture of \cite{kawamura2015simple}:

\medskip
\begin{theorem}\label{thm:fence schedule}
For any sufficiently large $k$, there exists a fence patrolling schedule with efficiency $1- \frac{3.5}{\sqrt{k}}$. Such a schedule uses $k-1$ agents of speed one and one agent with maximum speed $\Theta(\sqrt{k})$.
\end{theorem}

Note that this theorem implies that for any $\eps > 0$ there exists a fence patrolling schedule with efficiency $1 - \eps$ using $O(\frac{1}{\eps^2})$ agents -- one with speed $\Theta(\frac{1}{\eps})$ and all others with speed $1$. In other words, the efficiency can be made arbitrarily close to $1$ by choosing the appropriate number and maximum speeds of agents.

We remark that Theorem~\ref{thm:fence schedule} also shows that both our upper bounds are asymptotically tight. In particular, the optimal efficiency for any schedule with $k$ agents is indeed $1 - \Theta(\frac{1}{\sqrt{k}})$. Furthermore, for any $s\geq 1$, there is a configuration (with $k=\Theta(s^2)$ agents), where the maximum speeds of the agents differ by a factor $s$ and for which the optimal efficiency is $\frac{1}{1 + \frac{1}{\Theta(s)}} = 1 - \Theta(\frac{1}{s})$.

\subsection{Results for the Circle}

We resolve the conjecture by Kawamura and Soejima affirmatively. Namely, for any large enough $k$, we can construct a patrol schedule with idle time $1$ using agents with maximum speeds $1, 1/2, \dots 1/k$ that patrols a unidirectional circle of length $L=\Theta\left(\frac{\log_2 k}{\log_e \log_2 k}\right)$. In fact, our construction extends to a new universal scheme for the unidirectional circle. This is captured in the following theorem.

\begin{theorem}\label{thm:circle schedule}
For $k$ sufficiently large and for any $k$ agents with maximum speeds $v_1, \dots v_k$ there exists a patrol scheme with idle time $1$ that patrols a unidirectional circle of length
$$L = \frac{1}{33\log_e\log_2 k} \sum_{i=1}^k v_i.$$
\end{theorem}

\subsection{Organization}

The rest of the paper is organized as follows: We first give a more formal model description of the fence patrolling problem as well as discuss some related models and works in Section~\ref{sec:model}. In Section~\ref{sec:impossiblity} we explain and prove our upper bounds for the line segment. Section~\ref{sec:schedule} explains and gives a proof of our optimal fence patrolling schedule for the line segment. Finally, in Section~\ref{sec:circle_schedule} we present our schedule for a circle with length $\Theta(\frac{1}{\log_{e}\log_{2}{k}} \sum_{i=1}^k v_i)$ and proof that it indeed has idle time 1. Formal and complete proofs for the two schedules can be found in the appendix.

\section{Fence Patrolling and Related Models} \label{sec:model}

In this section we give a more detailed formal definition for the fence patrolling model/problem and briefly discuss related models and results.
The fence patrolling model as given by \cite{czyzowicz2011boundary} is defined as follows:

\begin{itemize}
\item The \emph{environment} $\mathcal{E}$ to be patrolled is $1$-dimensional and consists of a line segment of length $L$ or a circle of circumference $L$. This line segment or circle is also referred to as a \emph{fence}.
\item  The fence patrolling problem consists of some finite number $k \in \mathbb{N}$ of mobile agents $a_1, a_2, \ldots, a_k$ to patrol the fence, each having a possibly distinct positive maximum speed $v_1, v_2, \ldots, v_k \in \mathbb{R}_{+}$. 
\item A \emph{schedule} for the fence patrolling problem consists of a $k$-tuple of functions $a_1, a_2, \dots, a_k:[0, \infty)\rightarrow\mathcal{E}$ such that, for all $i\in[k]$, $t\geq 0$ and $\epsilon>0$,
$$\dist(a_i(t+\epsilon), a_i(t)) \leq \epsilon\cdot v_i.$$ That is, we assume patrolling starts at $t=0$ and goes on indefinitely. Each agent follows a predetermined trajectory, in which he/she moves along $\mathcal{E}$ with at most his/her maximum speed. In the case of a circular fence, the function $\dist(x,y)$ refers to the length of the shorter circle arc between $x,y \in \mathcal{E}$. In the case of the unidirectional circle, we have the additional requirement that $\forall i \in [k], t \geq 0$ and $0<\epsilon<\frac{L}{2v_i}$, the shorter arc between $a_i(t)$ and $a_i(t+\epsilon)$ is the one that spans clockwise from $a_i(t)$.
\item We say that a patrol schedule has \emph{idle time $T$} for some fixed positive parameter $T$ if for all $t\geq T$ and for all $x\in\mathcal{E}$, there is some agent that visits $x$ during $[t-T, t]$. Intuitively, this condition means that an intruder cannot remain undetected at a point for more than $T$ time.

\item  Given a patrol schedule, we say that a point $(x,t) \in \mathcal{E} \times [T,\infty)$ is \emph{$T$-covered} if some agent $a_i$ visits the point $x \in \mathcal{E}$ on the fence during the time interval $[t-T,t]$. Note that in this model an agent patrols/monitors a point  $x \in \mathcal{E}$ by visiting it. On the one hand, this means the agents are limited to zero line of sight. On the other hand, no additional operation (e.g. stop and look around) is necessary to patrol a point.
\end{itemize}

One can see that a schedule has idle time $T$ if and only if every point $(x,t) \in \mathcal{E} \times [T,\infty)$ is $T$-covered. It is easy to observe that any patrol schedule of a fence of length $L$ with idle time $T$ can be rescaled to a schedule of a fence of length $\alpha\cdot L$ with idle time $\frac{1}{\alpha}\cdot T$ for any $\alpha >0$. Thus, to simplify terminology, we assume henceforth that $T=1$ and we refer to $1$-covered simply as covered.

Related models have been considered in the literature: where agents have positive line of sight \cite{czyzowicz2014visibility}, where agents have distinct walking and patrolling speeds \cite{czyzowicz2017twospeed}, where some agents may be faulty \cite{czyzowicz2017faulty}, where only some regions of the environment need to be patrolled \cite{collins2013}, or where the environment is a geometric tree \cite{czyzowicz2017trees}. However, all of these models feature identical agents and in particular do not allow for varying maximum speeds. Overall, the model given above is likely the cleanest and most natural model in which agents with different speeds can and have been studied. Despite the extreme simplicity of this model, this paper and prior works on the fence patrolling problem~\cite{czyzowicz2011boundary,kawamura2015fence,dumitrescu2014fence,kawamura2015simple} show that very surprising and intricate phenomena occur when agents have different speeds and that these nontrivial consequences can be studied in the model defined above.

\section{Impossibility Results for the Line Segment: Proof of Theorem \ref{thm:GeneralUpperBoundFence} and Lemma \ref{lem:UpperBoundNumber}}\label{sec:impossiblity}

In this section, we prove two upper bounds on the length of a straight line fence (i.e. $\mathcal{E} = [0,L]$) patrolled by agents of maximum speeds $v_1, \dots, v_k$. 

\begin{proof}[\bf Proof of Theorem \ref{thm:GeneralUpperBoundFence}]
The main idea of the proof is to consider the two-dimensional spacetime continuum $\mathcal{S}:=[0,L] \times [0,\infty)$ and the trajectories of the agents along with the points they cover as geometric objects in it. To prove an upper bound on $L$, we add the agents one by one in a carefully chosen order. Whenever we add an agent, some additional points are covered by him/her and we then can examine what happens to the ``right border'' of what has been covered so far as we add more agents (see Figure~\ref{fig:fence impossibility border} below) to derive our results. 

\begin{figure}[!ht]
\begin{center}
\begin{tikzpicture}[scale=0.3]
\def \n {4}

\draw[fill=black] (0, 2*\n+0.6) rectangle (4*\n, 2*\n+1.1);
\node[above] at (2*\n, 2*\n+1.1) {fence};

\draw[black!20] (-0.5, 2*\n) -- (4*\n+0.5, 2*\n);
\node[left] at (-0.5, 2*\n) {0};

\draw[ultra thick, black!50, ->] (0, 2*\n) -- (0, -2*\n); 
\draw[ultra thick, black!50, ->] (4*\n, 2*\n) -- (4*\n, -2*\n); 
\node[left] at (0, -2*\n+0.5) {time};

\draw[fill=black!20, black!20] (3.2*\n, 2*\n) --(3.2*\n, 2*\n-1) -- (0.8*\n,-2*\n -1) -- (0.8*\n,-2*\n ) -- (3.2*\n, 2*\n);

\draw[fill=black!20, black!20] (2/3*\n,2*\n) -- (2/3*\n,2*\n-1) -- (0, 1.5*\n-1) -- (0, 1.5*\n) -- (2/3*\n,2*\n);
\draw[fill=black!20, black!20] (0, 1.5*\n) -- (0, 1.5*\n-1) -- (2.5*\n,-3/8*\n-1) -- (2.5*\n,-3/8*\n) --(0, 1.5*\n);
\draw[fill=black!20, black!20] (2.5*\n,-3/8*\n) -- (2.5*\n,-3/8*\n-1) -- (2*\n, -3/4*\n-1) -- (2*\n, -3/4*\n) -- (2.5*\n,-3/8*\n);
\draw[fill=black!20, black!20] (2*\n, -3/4*\n) -- (2*\n, -3/4*\n-1) --(11/3*\n, -2*\n-1) -- (11/3*\n, -2*\n) --(2*\n, -3/4*\n);

\draw[fill=black!20, black!20] (2*\n-2,2*\n) -- (2*\n-2,2*\n-1)-- (4*\n, \n-2) --  (4*\n, \n-1) -- (2*\n-2,2*\n); 
\draw[fill=black!20, black!20] (4*\n, 1*\n-1) -- (4*\n, 1*\n-2) -- (0,-\n-2)-- (0,-\n-1) -- (4*\n, \n-1) ;
\draw[fill=black!20, black!20]  (0,-\n-1) -- (0,-\n-2) --  (2*\n-2,-2*\n-1)  -- (2*\n-2,-2*\n) -- (0,-\n-1) ;

\draw[ultra thick, red] (3.2*\n, 2*\n) -- (0.8*\n,-2*\n);

\draw[ultra thick, blue]  (2/3*\n,2*\n)  -- (0, 1.5*\n) -- (2.5*\n,-3/8*\n) -- (2*\n, -3/4*\n)
 --(11/3*\n, -2*\n);

\draw[ultra thick, green] (2*\n-2,2*\n) -- (4*\n, \n-1) -- (0,-\n-1) -- (2*\n-2,-2*\n);

\draw[dashed, ultra thick, yellow] (3.2*\n, 2*\n-1) -- (2*\n+12/29,-9/29) --(2.5*\n, -3/8*\n) --  (2.5*\n, -3/8*\n-1) -- (2*\n+2/3, -3/4*\n -0.5) -- (11/3*\n, -2*\n);

\draw[dashed, ultra thick] (3.2*\n, 1.97*\n-1) -- (2.94*\n, 1.52*\n-1) -- (4*\n, 1*\n-1) -- (4*\n, 1*\n-2) -- (2.5*\n - 0.4+0.2, -3/8*\n + 0.3) --(2.5*\n+0.1,-3/8*\n) -- (2.5*\n+0.1, -3/8*\n-1) -- (2*\n+2/3+0.2, -3/4*\n -0.5) -- (11/3*\n+0.2, -2*\n);

\end{tikzpicture}
\caption{If we have added the blue agent and the red agent so far, then the right border is shown in yellow. If we now also add the green agent, the new right border is shown in black.}
\label{fig:fence impossibility border}
\end{center}
\end{figure}
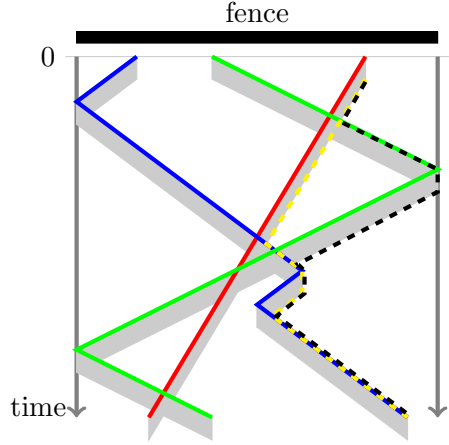

Since, as noted in Section \ref{sec:model}, a patrol schedule with agents $a_1, \dots, a_k$ has idle time $1$ if and only if $\forall x \in [0,L], \forall t \geq 1,$ $(x,t) \in \mathcal{S}$ is covered by at least one agent $a_j$, the theorem can be equivalently stated as that, for any patrol schedule such that all points $(x, t)\in S$ with $t\geq 1$ are covered by some agent, we have
\begin{equation} 
L\leq \sum_{i=1}^{k} \frac{1}{\frac{1}{v_i}+\frac{1}{v_{\max}}}. \label{eq::1} 
\end{equation}
In fact, we will show \eqref{eq::1} under the weaker assumption that only points $(x, t)\in S$ with $t\in [1, 2k]$ are covered by some agent.

Given a patrol schedule of $[0,L]$ with agents $a_1, \dots, a_k$ and a non-empty subset $A\subseteq \{a_1,\dots,a_k\}$, we define the \emph{right border} of $A$ as the function $B^A:[1, \infty)\rightarrow[0, L]$ given by 
$$B^A(t) := \max\{x \in [0,L] : (x,t)\textrm{ is covered by some agent in }A\}.$$
We show \eqref{eq::1} by considering, for some $1 \leq q \leq k$ (as we might not need to consider all agents), the collections of agents $A_1, \dots, A_q$, where $A_1 = \{a_{i_1}\}$, and for all $j \in [q]\setminus\{1\}, A_j = A_{j-1}\cup \{a_{i_j}\}$, $i_1, \dots, i_q$ is a sequence of distinct integers in $[k]$ to be specified later.  The intuition behind this is that we are starting with an empty set and adding more agents in a specific order until some termination condition is met. It is clear that for all $t \geq 0$ and for all  $j \in [q-1]$ we have $B^{A_j}(t)\leq B^{A_{j+1}}(t)$. The key idea of the proof is to consider what happens to the right border of $A_j$ as $j$ increases (that is, as more agents are added). An example of the right borders of $A_j$ and $A_{j+1}$ for some $j$ is shown in Figure \ref{fig:fence impossibility border}.

At this point we prove a claim which will be useful in specifying the sequence $i_1, \dots, i_q$.
\begin{claim}
\label{existenceOfSingleAgent}
For any patrol schedule of $[0,L]$ with set of agents $A = \{a_1, \dots, a_k\}$ and idle time $1$, for any subset $A'$ of $\{a_1, \dots, a_k\}$, and for any point $(p_x, p_t) \in \mathcal{S}$ on the right border of $A'$ (that is, such that $B^{A'}(p_t) = p_x$), there exists an $\varepsilon > 0$ such that there is at least one agent $a_i \in A \setminus A'$ that covers all points $(p_x+\nu, p_t)$ for $\nu \in [0,\varepsilon]$.
\end{claim}
\begin{proof}
Note that there could not exist three points $(x_1, t), (x_2, t), (x_3, t)$ such that $0 \leq x_1 < x_2 < x_3 \leq L$ and some agent $a_j$ covers $(x_1, t)$ and $(x_3, t)$ but not $(x_2, t)$. This is because the trajectory of every agent can be considered as a continuous function $f_{a_j}:[0,\infty) \rightarrow [0,L]$, so it cannot be that $\exists t_1,t_3 \in [t-1,t]$ such that $f(t_1) = x_1$ and $f(t_3) = x_3$ but $\not\exists t_2 \in [t-1,t]$ with $f(t_2) = x_2$. It follows that the set of points $C_j$ on the segment between $(p_x, p_t)$ and $(L, p_t)$ covered by some agent $a_j$ must be either the empty set or a segment. Now consider the set of agents $A_p$ in $A \setminus A'$ that cover $(p_x, p_t)$ and note that $\exists \varepsilon > 0$ such that $\exists a \in A_p$ that covers $(p_x + \varepsilon, p_t)$ (otherwise choose the non-empty $C_j$ with $a_j \in A \setminus A'$ with the leftmost left end $(p_{l},p_t)$ that is strictly to the right of $p_x$ and notice that all points $(p_{x'},p_t)$ with $p_{x'} \in (p_x, p_l)$ are not covered by any agent in $A$, which is impossible as all points in $\mathcal{S}$ should be covered). But now $a \in A_p$ covers both $(p_x, p_t)$ and $(p_x + \varepsilon, p_t)$, therefore it also covers anything in between since the points it covers between $(p_x,p_t)$ and $(L, p_t)$ must form a segment.
\end{proof}

Now we can continue with the proof of our main upper bound by specifying $i_1, \dots, i_q$. Consider a fixed patrol schedule of $[0,L]$ with agents $a_1, \dots, a_k$ and idle time $1$. To pick the sequence $i_1, \dots, i_q$, consider the following procedure: initially, put $l_0=(x_0, t_0) = (0, k)$ and pick $i_1\in[k]$ such that agent $a_{i_1}$ covers $l_0$. For each consecutive $j=1, 2, \dots$, we let $l_{j} = (x_{j},t_{j})$ be such that
$$t_{j}=\arg\min_{t \in [k-j, k+j]} B^{A_{j}}(t)$$
and $x_{j} = B^{A_{j}}(t_{j})$. Intuitively, $l_{j}$ is the leftmost point of $B^{A_{j}}$ between times $k-j$ and $k+j$. Now if $x_j = L$, we stop adding agents and we set $q:= j$. Note that if $j=k$, then $x_j = L$ as agents $a_1, \dots, a_k$ cover all of $[1,2k]$. If $x_j < L$, pick an agent $a_{i_{j+1}}$ that covers all the points with coordinates $(x_{j}+\nu, t_{j})$ for $\nu \in [0,\varepsilon]$ for some small enough $\varepsilon > 0$. Such an agent should exist by Claim \ref{existenceOfSingleAgent}. To make sure $l_j$ is always defined for any $j \in \{0,1,\dots,q\},$ if $q=k$, set $l_k := (L, k)$.

We note that $x_0, x_1, ..., x_q$ is non-decreasing, $x_0=0$ and $x_q=L$ since we either stopped adding agents when $q < k$ because $x_q=L$ or we stopped when $q=k$, in which case all of $[0,L] \times [1,2k]$ should be covered. Hence the theorem follows if we can show that
$$x_{j+1} - x_{j} \leq \frac{1}{\frac{1}{v_{i_{j+1}}} + \frac{1}{v_{\max}}}$$
is true for all $ j \in \{0,1,\dots,q-1\}$.

In order to bound this difference, we investigate how the right border moves when agent $a_{i_{j+1}}$ is added. Note that $\forall t \in [0,\infty),$ $$B^{A_{j+1}}(t) = \max(B^{A_j}(t), B^{\{a_{j+1}\}}(t)).$$ For any time $t > t_{j}$, the rightmost point at time $t$ that could be covered by any agent in $A_{j}$ is $(x_{j} + (t - t_{j})v_{\max},t)$ since the speed of any agent is at most $v_{\max}$. Similarly, for any time $t < t_{j}$, the rightmost point at time $t$ that could be covered is $(x_{j} + (t_{j} - t)v_{\max},t)$. Thus, $\forall t \in [0,\infty),$ $$B^{A_{j}}(t) \leq x_{j}+\lvert t-t_{j}\rvert v_{\max}.$$Denote by $u$ the ray $(x_{j} + (t_{j} - t)v_{\max},t)$ where $t \leq t_{j}$, and by $w$ the ray $(x_{j} + (t - t_{j})v_{\max},t)$ where $t \geq t_{j}$.

Next, consider $B^{\{a_{j+1}\}}(t)$. Since agent $a_{j+1}$ covers $l_{j}$, this means that he/she visits $x_j$ at some time between $t_{j}-1$ and $t_j$, say at point $(x_{j}, t_{visit})$. Under the restriction that the trajectory of agent $a_{i_{j+1}}$ should go through $(x_j, t_{visit})$, it is clear that $B^{\{a_{i_{j+1}}\}}(t)$ is maximized if agent $a_{i_{j+1}}$ comes from the right at maximum speed, hits $(x_{j}, t_{visit})$ and then turns around and moves to the right at maximum speed, in which case equality is achieved in $$\forall t \in [0,\infty), B^{\{a_{i_{j+1}}\}}(t) \leq \min(x' + \lvert t - t' \rvert v_{i_{j+1}},L),$$where $(x',t') = (x_{j} + \frac{v_{i_{j+1}}}{2}, t_{visit} + \frac{1}{2})$. Denote by $h$ the ray $(x' + (t' - t)v_{i_{j+1}},t)$ where $t \leq t'$ and by $g$ the ray $(x' + (t - t')v_{i_{j+1}},t)$ where $t \geq t'$.

\begin{figure}[!ht]
\begin{center}
\begin{tikzpicture}[scale=0.8]
\def \n {3}
\foreach \x in {0,2*\n,4*\n} {

\draw[fill=black!20, black!20] (\x+\n, \n) -- (\x+\n, \n -2) -- (\x +1, -1) --
(\x+\n, -\n) -- (\x+\n, -\n-2) --(\x,-2) --(\x,0) --(\x+\n, \n);

\draw[ultra thick] (\x+\n, \n) -- (\x,0) -- (\x+\n, -\n);
}

\draw[ultra thick, red] (0, -0.5) -- (\n,0);
\draw[ultra thick, green] (0, -0.5) -- (\n,-1);

\node[left] at (0,0) {$N := L'$};
\draw[fill] (0,0) circle (0.1cm);

\node[left] at (0,-0.5) {$L_{\text{old}}$};
\draw[fill] (0,-0.5) circle (0.1cm);

\node[above right] at (1.28,-0.72) {$L_{\text{new}}$};
\draw[fill] (1.28,-0.72) circle (0.1cm);

\node[right] at (1.28,1.28) {$M$};
\draw[fill] (1.28,1.28) circle (0.1cm);
\draw[dashed, ultra thick] (1.28,-0.72) -- (1.28,1.28);

\draw[ultra thick, red] (2*\n, -1.5) -- (3*\n,-1);
\draw[ultra thick, green] (2*\n, -1.5) -- (3*\n,-2);

\node[left] at (2*\n, -1.5) {$L_{\text{old}}$};
\draw[fill] (2*\n, -1.5) circle (0.1cm);

\node[above right] at (2*\n+1.25,-1.25) {$L_{\text{new}}$};
\draw[fill] (2*\n+1.25,-1.28) circle (0.1cm);

\node[right] at (2*\n+1.25,-3.28) {$M$};
\draw[fill] (2*\n+1.25,-3.28) circle (0.1cm);
\draw[dashed, ultra thick] (2*\n+1.25,-1.28) -- (2*\n+1.25,-3.28);

\node[left] at (2*\n,0) {$L'$};
\draw[fill] (2*\n,0) circle (0.1cm);

\node[left] at (2*\n,-2) {$N$};
\draw[fill] (2*\n,-2) circle (0.1cm);

\draw[ultra thick, red] (4*\n, -1.5) -- (5*\n, 1.2);
\draw[ultra thick, green] (4*\n, -1.5) -- (5*\n,-4.2);

\node[left] at (4*\n, -1.5) {$L_{\text{old}}$};
\draw[fill] (4*\n, -1.5) circle (0.1cm);

\node[right] at (4*\n+1,-1) {$L_{\text{new}}$};
\draw[fill] (4*\n+1,-1) circle (0.1cm);

\node[left] at (4*\n,0) {$N:= L'$};
\draw[fill] (4*\n,0) circle (0.1cm);

\node[left] at (4*\n+1,1) {$M$};
\draw[fill] (4*\n+1,1) circle (0.1cm);
\draw[dashed, ultra thick] (4*\n+1,1) -- (4*\n+1,-1);

\end{tikzpicture}
\caption{The rays $R_1$ and $R_2$ give a bound to the right of the current right border and are given in red and green respectively, and in black and gray we have a bound to the right of the trajectory of the newly added agent $a_i$ and its shadow, i.e. the points that are covered by it.}
\label{fig:fence impossibility cases}
\end{center}
\end{figure}
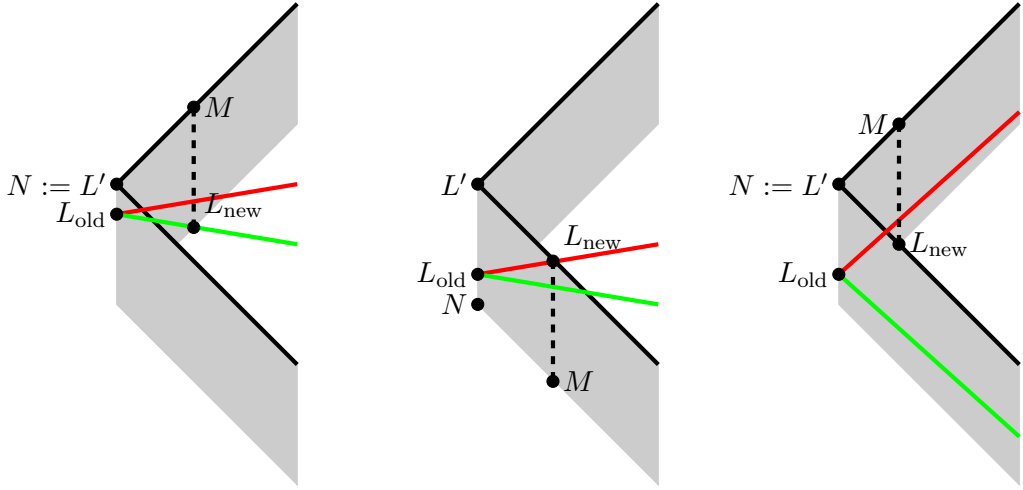 

We thus get $\forall t \in [0,\infty),$ $$B^{A_{j+1}}(t) = \max(  B^{A_{j}}(t), B^{\{a_{i_{j+1}}\}}(t) ) \leq f(t),$$ where $f(t) = \max(x_{j}+\lvert t-t_{j} \rvert v_{\max}, x' + \vert t-t' \rvert v_{i_{j+1}})$. Consider $t_{new} = \arg\min_{t \in [1,\infty)} f(t)$. Let $L_{new} = (x_{new} := f(t_{new}),t_{new})$ be the leftmost point on the aforementioned upper bound on $B^{A_{j+1}}(t)$. Notice that $L_{new}$ is either the intersection of $h$ and $w$, or the intersection of $g$ and $u$, or the intersection of $g$ and $h$. These three cases are illustrated in Figure \ref{fig:fence impossibility cases}. We have $u$ in red and $w$ in green. Consider the upper bound on $B^{\{a_{i_{j+1}}\}}(t)$ mentioned above. The trajectory of agent $a_{i_{j+1}}$ that would correspond to matching this upper bound is given in black and the points $a_{i_{j+1}}$ would cover if this was his/her trajectory are given in gray. We have that $L_{old} = (x_{old}, t_{old}) := l_{j}$. It can be seen by inspection of the three cases in Figure \ref{fig:fence impossibility cases} that $|t_{new} - t_{j}| \leq 1$. Then $t_{new} \in [k-(j+1),k+(j+1)]$, which makes $L_{new}$ a candidate for $l_{j+1}$, therefore $l_{j+1} = (x_{j+1}, t_{j+1})$ will have $x_{j+1} \leq x_{new}$. Thus it is enough to show that $x_{new} - x_j \leq \frac{1}{\frac{1}{v_{i_{j+1}}} + \frac{1}{v_{\max}}}$.

We need an upper bound on $d = x_{new} - x_{old}$. We consider the points $M = (x_M, t_M)$ and $N = (x_N, t_N)$ as illustrated in Figure \ref{fig:fence impossibility cases}, such that in all three cases $x_M = x_{new}$ and $\lvert t_M - t_{new} \rvert=1$. In Case 1, we consider the segment $MN$ of slope $\frac{1}{v_{i_{j+1}}}$ and $x_M - x_N = d$, and the segment $L_{old}L_{new}$ of $w$ of slope $-\frac{1}{v_{\max}}$ and $x_{new} - x_{old} = d$. This gives us
\begin{equation}
\label{eq::2}
\frac{d}{v_{i_{j+1}}} + \frac{d}{v_{\max}} \leq 1 \Rightarrow d \leq \frac{1}{ \frac{1}{v_{i_{j+1}}} + \frac{1}{v_{\max}} }.
\end{equation}
In Case 2, we consider the segment $L_{new}L_{old}$ of $u$ of slope $\frac{1}{v_{\max}}$ and $x_{new} - x_{old} = d$, and the segment $NM$ of slope $-\frac{1}{v_{i_{j+1}}}$ and $x_{M} - x_{N} = d$. This implies Equation \eqref{eq::2} for Case 2 as well. In Case 3, we consider the segment $MN$ of slope $\frac{1}{v_{i_{j+1}}}$ and $x_M - x_N = d$, and the segment $NL_{new}$ of slope $-\frac{1}{v_{i_{j+1}}}$ and $x_{new} - x_{N} = d$. This means that $$\frac{2d}{v_{i_{j+1}}} = 1 \Rightarrow d = \frac{v_{i_{j+1}}}{2} \leq \frac{1}{ \frac{1}{v_{i_{j+1}}} + \frac{1}{v_{\max}} }.$$Therefore, $x_{new} - x_{old} \leq \frac{1}{\frac{1}{v_{i_{j+1}}} + \frac{1}{v_{\max}}}$ as desired.
\end{proof}

\begin{proof}[\bf Proof of Lemma \ref{lem:UpperBoundNumber}]
We show how $L\leq \left(1-\frac{1}{\sqrt{k}+1}\right)\sum_{i=1}^k v_i$ follows from Theorem~\ref{thm:GeneralUpperBoundFence}. First note that $L \leq \sum_{i=1}^k v_i - \frac{v_{\max}}{2}$ as each agent $a_i$ contributes at most $v_i \cdot \frac{1}{1 + \frac{v_i}{v_{\max}}} \leq v_i$ while the agent with maximum speed contributes exactly $\frac{v_{\max}}{2}$. Therefore, if $v_{\max} \geq \frac{1}{\sqrt{k}} \sum_{i=1}^k v_i$ the desired upper bound for $L$ follows immediately. It remains to deal with the case $v_{\max} < \frac{1}{\sqrt{k}} \sum_{i=1}^k v_i$. For this we first note that $x\cdot \frac{1}{1 + \frac{x}{v_{\max}}} = \frac{1}{\frac{1}{x} + \frac{1}{v_{\max}}}$ is a concave function in $x$ for $0 \leq x \leq v_{\max}$, since the second derivative $-\frac{2}{(1+\frac{x}{v_{\max}})^3v_{\max}}$ is always negative. This allows us to apply Jensen's inequality and thus we have
\begin{align*}
L \leq \sum_{i=1}^k \frac{1}{ \frac{1}{v_i} + \frac{1}{v_{\max}} }
 \leq k \frac{1}{\frac{1}{v_{avg}} + \frac{1}{v_{\max}} } = \frac{1}{1 + \frac{\sum_{i=1}^k v_i}{k \cdot v_{\max}}} \sum_{i=1}^k v_i \leq \left(1 - \frac{1}{\sqrt{k}+1}\right) \sum_{i=1}^k v_i,
\end{align*}
where $v_{avg} = \frac{1}{k} \sum_{i=1}^k v_i$. This concludes the proof of Lemma \ref{lem:UpperBoundNumber}.
\end{proof}

\section{A Schedule with Efficiency $1 - \eps$ for the Line Segment: Proof of Theorem \ref{thm:fence schedule}}\label{sec:schedule}
In this section, we prove that for any $k$ agents, there exist speeds $v_1, \dots, v_k$ and a scheme for these agents to patrol a fence of length 
$$L = \left(1- \frac{3.5}{\sqrt{k}} + O(1/k) \right) \sum_{i=1}^k v_i.$$
This improves the result from \cite{kawamura2015simple} and therewith falsifies the corresponding conjecture stated in that paper.

\begin{proof}[Proof of Theorem \ref{thm:fence schedule}]
Assume $k$ is sufficiently large, and, for ease of notation, define $n:=k-2$. Let $L = n - 3/2\sqrt{n}$. We construct a schedule that patrols $\mathcal{E}=[0, L]$ with idle time 1, using $n+1$ agents with maximum speed $1$ and $1$ agent with maximum speed $2\sqrt{n} -1$. Thus we have a total speed of $V = \sum_{i=1}^{k} v_i = n+2\sqrt{n}$. As the ratio between $L$ and $V$ approaches $1- \frac{3.5}{\sqrt{k}} + O(1/k)$, Theorem \ref{thm:fence schedule} follows.

To simplify presentation of the patrol schedule, we will allow agents to occasionally  ``step out of the fence $[0, L]$'', i.e. we allow an agent $a_i$ to assume positions $a_i(t)<0$ and $a_i(t)>L$ (to avoid this, we could also modify the schedule so that they stay at the respective end of the fence for a while). To keep the notation as clean as possible, we henceforth assume that $n$ is a square number. Our schedule works as follows (see figure Figure \ref{fig:schedule} for a graphical representation):
\bigskip

\noindent \textsc{Slow agents $a_1, \ldots, a_n$:} 
For each $i \in \{0, \ldots, n\}$, agent $a_i$ starts at time $0$ at position $x=  i - i / \sqrt{n}$ and moves $i / (2\sqrt{n})$ time units to the right. Then he or she repeats:
\begin{itemize}
\item move to the left for $\sqrt{n}$ time units.
\item move to the right for $\sqrt{n}$ time units.
\end{itemize}

\noindent \textsc{Fast agent $a_{n+1}$:}
The fast agent $a_{n+1}$ starts at time $0$ and repeats the following four steps:
\begin{enumerate}[(1)]
\item Move from position $0$ to position $L+1/2$ with speed $2\sqrt{n} -1$.
\item Move from position $L+1/2$ to position $-1/2$ during the next $\sqrt{n}/2 +1$ time units (e.g. with constant speed $(L+1)/(\sqrt{n}/2 +1) = 2\sqrt{n} -7+  16/(\sqrt{n}+2)$).
\item Move from position $-1/2$ to position $L$ with speed $2\sqrt{n} -1$.
\item Move from position $L$ to position $0$ in the next $\sqrt{n}/2$ time units (e.g. with constant speed $L/(\sqrt{n}/2) = 2\sqrt{n} -3$).
\end{enumerate}

The idea behind our patrol schedule is to initially place the agents with maximum speed $1$ equidistantly along the fence with gaps of length slightly smaller than $1$, similar to the schedule for the fast agents in \cite{kawamura2015simple}. In contrast to their schedule, this is performed slightly out of phase between the agents. This will cover most of the points on the fence. The only problem appears whenever the agents turn around, as then the points right next to these turning points are not visited for more than 1 time unit, hence creating uncovered triangles in the ``spacetime'' diagram (white triangles in Figure \ref{fig:schedule}). By timing the turning times of the agents appropriately, we ensure that these uncovered triangles are placed such that they can all be cleaned up by the last fast agent. This will be described in further detail in the paragraphs that follow. Figure~\ref{fig:schedule} gives a complete illustration of our schedule. 

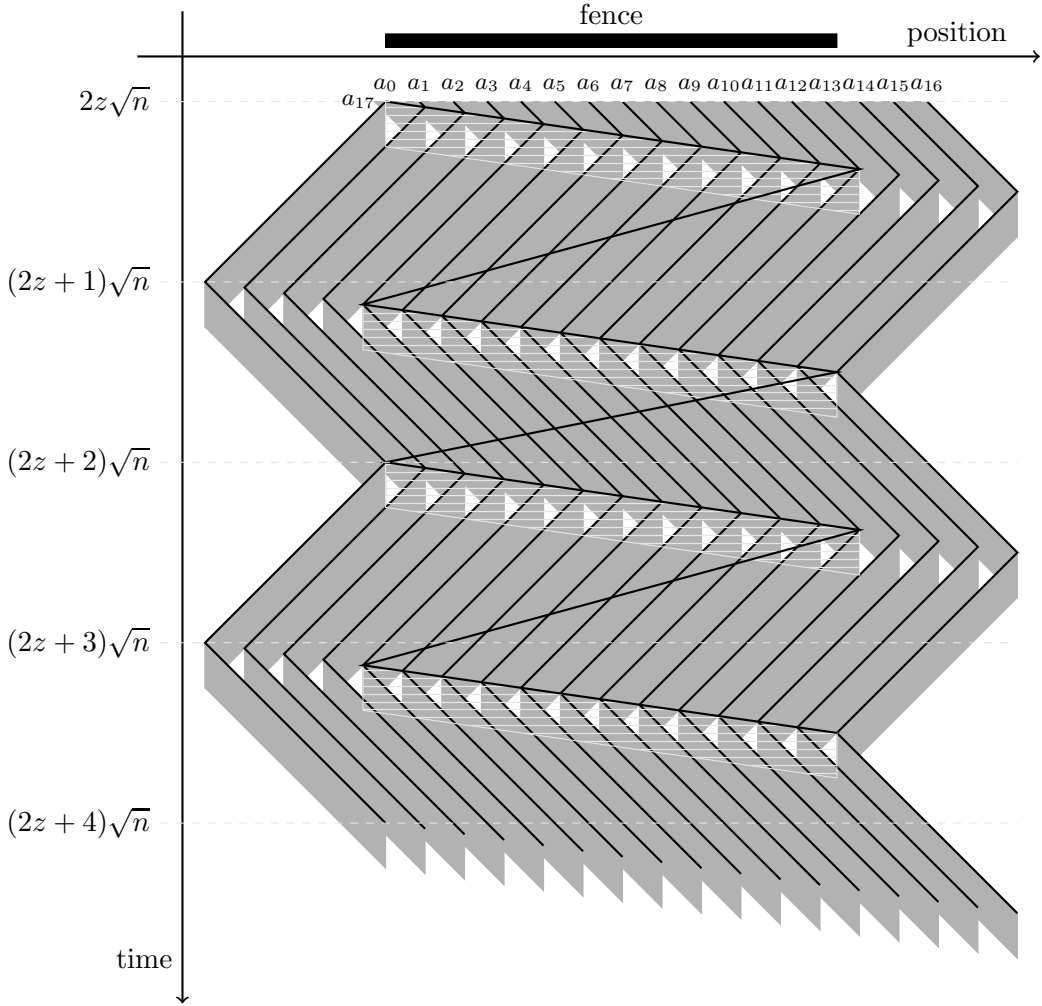
\begin{figure}[!ht]
\begin{center}
\begin{tikzpicture}[scale=0.6]

\def \n {16}
\def \sqrtn {4}

\foreach \x in {1,...,\n}{
\draw[fill=black!30,black!30]  (\x -\x/\sqrtn ,0) -- (\x - \x*0.5/\sqrtn , - \x*0.5/\sqrtn)
-- (\x - \x*0.5/\sqrtn , - \x*0.5/\sqrtn-1) -- (\x -\x/\sqrtn -1, 0) --(\x -\x/\sqrtn ,0); 
} 

\foreach \x in {0,...,\n}{
\draw[fill=black!30,black!30]  (\x - \x*0.5/\sqrtn , - \x*0.5/\sqrtn) --
(\x - \x*0.5/\sqrtn -\sqrtn, - \x*0.5/\sqrtn -\sqrtn) -- 
(\x - \x*0.5/\sqrtn -\sqrtn, - \x*0.5/\sqrtn -\sqrtn -1) --
(\x - \x*0.5/\sqrtn , - \x*0.5/\sqrtn-1) --
(\x - \x*0.5/\sqrtn , - \x*0.5/\sqrtn); 
}

\foreach \x in {0,...,\n}{
\draw[fill=black!30,black!30]  (\x - \x*0.5/\sqrtn -\sqrtn, - \x*0.5/\sqrtn -\sqrtn) --
(\x - \x*0.5/\sqrtn, - \x*0.5/\sqrtn -2*\sqrtn) --
(\x - \x*0.5/\sqrtn, - \x*0.5/\sqrtn -2*\sqrtn -1) --
(\x - \x*0.5/\sqrtn -\sqrtn, - \x*0.5/\sqrtn -\sqrtn-1) --
(\x - \x*0.5/\sqrtn -\sqrtn, - \x*0.5/\sqrtn -\sqrtn);
}

\foreach \x in {0,...,\n}{
\draw[fill=black!30,black!30]  (\x - \x*0.5/\sqrtn , - \x*0.5/\sqrtn -2*\sqrtn) --
(\x - \x*0.5/\sqrtn -\sqrtn, - \x*0.5/\sqrtn -3*\sqrtn) -- 
(\x - \x*0.5/\sqrtn -\sqrtn, - \x*0.5/\sqrtn -3*\sqrtn -1) --
(\x - \x*0.5/\sqrtn , - \x*0.5/\sqrtn - 2*\sqrtn -1) --
(\x - \x*0.5/\sqrtn , - \x*0.5/\sqrtn - 2*\sqrtn); 
}

\foreach \x in {0,...,\n}{
\draw[fill=black!30,black!30]  (\x - \x*0.5/\sqrtn -\sqrtn, - \x*0.5/\sqrtn -3*\sqrtn) --
(\x - \x*0.5/\sqrtn, - \x*0.5/\sqrtn -4*\sqrtn) --
(\x - \x*0.5/\sqrtn, - \x*0.5/\sqrtn -4*\sqrtn -1) --
(\x - \x*0.5/\sqrtn -\sqrtn, - \x*0.5/\sqrtn - 3*\sqrtn-1) --
(\x - \x*0.5/\sqrtn -\sqrtn, - \x*0.5/\sqrtn -3*\sqrtn);
} 


\foreach \x in {0,...,\n}{
\draw[thick]  (\x -\x/\sqrtn ,0) -- (\x - \x*0.5/\sqrtn , - \x*0.5/\sqrtn) --
(\x - \x*0.5/\sqrtn -\sqrtn, - \x*0.5/\sqrtn -\sqrtn) --
(\x - \x*0.5/\sqrtn, - \x*0.5/\sqrtn -2*\sqrtn) --
(\x - \x*0.5/\sqrtn -\sqrtn, - \x*0.5/\sqrtn -3*\sqrtn) --
(\x - \x*0.5/\sqrtn, - \x*0.5/\sqrtn -4*\sqrtn); 
} 

\draw[gray!20, schraffiert=gray!20]  (0,0) -- 
(\n - 1.5*\sqrtn +0.5, - 0.5*\sqrtn +0.5) -- 
(\n - 1.5*\sqrtn +0.5, - 0.5*\sqrtn -0.5) --
(0,-1) -- (0,0);

\draw[gray!20, schraffiert=gray!20]  (-0.5, - \sqrtn - 0.5) --  
( \n- 1.5*\sqrtn, - 1.5*\sqrtn) -- ( \n- 1.5*\sqrtn, - 1.5*\sqrtn -1) -- 
(-0.5, - \sqrtn - 1.5) --  (-0.5, - \sqrtn - 0.5);

 \draw[gray!20, schraffiert=gray!20]  (0,-2*\sqrtn) -- 
(\n - 1.5*\sqrtn +0.5, - 2.5*\sqrtn +0.5) -- 
(\n - 1.5*\sqrtn +0.5, - 2.5*\sqrtn -0.5) --
(0, -2*\sqrtn-1) -- (0, -2*\sqrtn);

\draw[gray!20, schraffiert=gray!20]  (-0.5, - 3*\sqrtn - 0.5) --  
( \n- 1.5*\sqrtn, - 3.5*\sqrtn) -- ( \n- 1.5*\sqrtn, - 3.5*\sqrtn -1) -- 
(-0.5, - 3*\sqrtn - 1.5) --  (-0.5, - 3*\sqrtn - 0.5);

\draw[thick]  (0,0) -- (\n - 1.5*\sqrtn +0.5, - 0.5*\sqrtn +0.5) --
 (-0.5, - \sqrtn - 0.5) --  ( \n- 1.5*\sqrtn, - 1.5*\sqrtn) -- (0, -2*\sqrtn) --
 (\n - 1.5*\sqrtn +0.5, - 2.5*\sqrtn +0.5) --
 (-0.5, - 3*\sqrtn - 0.5) --  ( \n- 1.5*\sqrtn, - 3.5*\sqrtn);

\draw[thick, ->] (-\sqrtn-1.5, 1) -- (\n-1.5,1);
\node[above left] at (\n-2,1) {position};

\draw[thick, ->] (-\sqrtn-0.5, 2) -- (-\sqrtn-0.5, -5*\sqrtn);
\node[left] at (-\sqrtn-0.5, -5*\sqrtn+1) {time};

\draw[fill=black] (0, 1.2) rectangle (\n -1.5*\sqrtn, 1.5);
\node[above] at (\n/2 - 0.75*\sqrtn, 1.5) {fence};

\draw[black!10, dashed] (-\sqrtn -1, 0) -- (\n -2, 0);
\node[left] at (-\sqrtn -1, 0) {$2z\sqrt{n}$};

\draw[black!10, dashed] (-\sqrtn -1, -\sqrtn) -- (\n -2, -\sqrtn);
\node[left] at (-\sqrtn -1, -\sqrtn) {$(2z +1)\sqrt{n}$};

\draw[black!10, dashed] (-\sqrtn -1, -2*\sqrtn) -- (\n -2, -2*\sqrtn);
\node[left] at (-\sqrtn -1, -2*\sqrtn) {$(2z+2)\sqrt{n}$};

\draw[black!10, dashed] (-\sqrtn -1, -3*\sqrtn) -- (\n -2, -3*\sqrtn);
\node[left] at (-\sqrtn -1, -3*\sqrtn) {$(2z +3)\sqrt{n}$};

\draw[black!10, dashed] (-\sqrtn -1, -4*\sqrtn) -- (\n -2, -4*\sqrtn);
\node[left] at (-\sqrtn -1, -4*\sqrtn) {$(2z+4)\sqrt{n}$};

\foreach \x in {0,...,\n}  {
\node[above] at (\x -\x/\sqrtn ,0) {\footnotesize $a_{\x}$};
}

\node[left] at (0 ,0) {\footnotesize $a_{17}$};


\end{tikzpicture}
\caption{The described schedule for $n=16$. The dark grey area describes the points $(x,t)$ which are covered by the slow agents $a_0, \ldots, a_n$ while the light grey shaded area describes the points $(x,t)$ which are covered by the fast agent in steps (1) and (3) of his protocol.} 
\label{fig:schedule}
\end{center}
\end{figure}

We will show that the above schedule indeed has idle time 1. As mentioned in Section \ref{sec:model}, this is equivalent to showing that every $(x,t)\in [0,L] \times [1,\infty)$ is covered.

Observe that in our schedule all agents have periodicity $2\sqrt{n}$ (agent $a_{n+1}$ walks $\sqrt{n}/2 -1/2$ time units in steps (1) and (3), $\sqrt{n}/2+1$ time units in step (2) and $\sqrt{n}/2$ time units in step (4)). Thus, any times $t$ below should be interpreted as $t \mod 2\sqrt{n}$ (for clarity of representation we will usually omit the  $\mod 2\sqrt{n}$ term and hope that this will not cause confusion). Moreover, due to this $2\sqrt{n}$-periodicity, we could easily extend our patrol schedule to a schedule for all times $t\in (-\infty, \infty)$, i.e. we need not worry about times $t$ being smaller than 1 in our arguments below.

Denote by $L^i = \left(x^i_L, t^i_L\right)$ and $R^i = \left(x^i_R, t^i_R\right)$ the left respectively right turning point of agent $a_i$, $\forall i \in \{0, \ldots, n\}$, i.e. agent $a_i$ walks along the fence from $x_L^i$ to $x_R^i$ and back, turning around at times $t^i_L$ and $t_R^i$ respectively. It follows directly from the protocol of the slow agents that 
\begin{align*}
R^i &= \left(i \left(1 - \frac{1}{2\sqrt{n}}\right), \frac{i}{2\sqrt{n}}  \right) \quad \text{ and} \\
L^i &= \left(i \left(1 - \frac{1}{2\sqrt{n}}\right) - \sqrt{n}, \frac{i}{2\sqrt{n}} + \sqrt{n} \right) 
\end{align*}

Let $x$ be a fixed point along the fence. We will now argue that the point $x$ gets visited at least every 1 time unit. To do this, we will write down a $2\sqrt{n}$-periodic sequence of \emph{visiting times}, (i.e. times $x$ gets visited by an agent $a_i$) such that any two neighbouring visiting times differ by at most 1 time unit.  

First, we note that agent $a_i$ visits $x$ if $i(1-1/(2\sqrt{n})) - \sqrt{n} \leq x \leq i(1-1/(2\sqrt{n}))$ or equivalently if 
$$\frac{2\sqrt{n}x}{2\sqrt{n} -1} \leq i \leq \frac{2\sqrt{n}x+2n}{2\sqrt{n} -1}.$$
Denote by $j_\text{min} := \left\lceil \frac{2\sqrt{n}x}{2\sqrt{n} -1} \right\rceil \geq \left\lceil \frac{2\sqrt{n}\cdot 0}{2\sqrt{n} -1} \right\rceil = 0$ and by $j_\text{max} := \left\lfloor \frac{2\sqrt{n}x+2n}{2\sqrt{n} -1} \right\rfloor \leq  \left\lfloor \frac{2\sqrt{n}L+2n}{2\sqrt{n} -1} \right\rfloor = n$ the indices such that agents $a_{j_{\text{min}}}, \ldots, a_{j_{\text{max}}}$ are exactly the agents visiting the point $x$.

For each $i \in \{j_{\text{min}}, \ldots, j_{\text{max}}\}$ the agent $a_i$ visits $x$ at times 
$$ s_i =  \frac{i}{2\sqrt{n}}  - \left(i\left(1 - \frac{1}{2\sqrt{n}}\right) - x \right) = x - i\left(1 - \frac{1}{\sqrt{n}}\right) $$
walking from left to right and
$$ t_i =  \frac{i}{2\sqrt{n}}  + \left(i\left(1 - \frac{1}{2\sqrt{n}}\right) - x \right) = i-x$$
walking from right to left.
Furthermore, we observe that the fast agent $a_{n+1}$ visits $x$ in steps (1) and (3) of his protocol at times
\begin{align*}
 f_1 &=  \frac{x}{2\sqrt{n}-1}   \quad \text{and} \\
 f_3 &= \frac{x+1/2}{2\sqrt{n}-1} + \sqrt{n} + \frac{1}{2} = \frac{x + 2n}{2\sqrt{n}-1}.
\end{align*}
We claim that 
$$s_{j_{\text{max}}}, \ldots, s_{j_{\text{min}}}, f_{1}, t_{j_{\text{min}}}, \ldots, t_{j_{\text{max}}}, f_3$$
is a sequence of visiting times of the point $x$ with all adjacent visiting times differing in at most 1 time unit. It is obvious that the differences $s_{i-1} - s_i$ and $t_i - t_{i-1}$ are not greater than 1 for all $i = j_{\text{min}} +1, \ldots j_{\text{max}}$. Thus it remains to check whether this is also true for the remaining four gaps between $s_{j_{\text{min}}}$ and $f_{1}$,  $f_{1}$ and $t_{j_{\text{min}}}$, $t_{j_{\text{max}}}$ and $f_3$ and $f_3$ and $s_{j_{\text{max}}} + 2\sqrt{n}$.

As $j_{\text{min}} \leq \frac{2\sqrt{n}x}{2\sqrt{n} -1} +1$ we have
\begin{align*}
f_1 - s_{j_{\text{min}}} &= \frac{x}{2\sqrt{n}-1} - \left(x - j_{\text{min}}\left(1 - \frac{1}{\sqrt{n}}\right)\right) \\
&\leq  \frac{x}{2\sqrt{n}-1} -x  + \left( \frac{2\sqrt{n}x}{2\sqrt{n}-1} +1 \right)\left(1 - \frac{1}{\sqrt{n}}\right) \\
&= \frac{x}{2\sqrt{n}-1} - \frac{2\sqrt{n}x -x}{2\sqrt{n}-1} + \frac{2\sqrt{n}x}{2\sqrt{n}-1} - \frac{2x}{2\sqrt{n}-1} + 1 - \frac{1}{\sqrt{n}} = 1 - \frac{1}{\sqrt{n}}
\end{align*}
and 
\begin{align*}
t_{j_{\text{min}}} -f_1 &=  j_{\text{min}} - x - \frac{x}{2\sqrt{n}-1}\\
&\leq \frac{2\sqrt{n}x}{2\sqrt{n}-1} +1 -x - \frac{x}{2\sqrt{n}-1} \\
&= \frac{2\sqrt{n}x}{2\sqrt{n}-1} - \frac{2\sqrt{n}x -x}{2\sqrt{n}-1} - \frac{x}{2\sqrt{n}-1} +1 = 1
\end{align*}

Likewise, as $j_{\text{max}} \geq \frac{2\sqrt{n}x +2n}{2\sqrt{n}-1} -1$, we have
\begin{align*}
f_3 - t_{j_{\text{max}}} &= \frac{x + 2n}{2\sqrt{n}-1} - \left(j_{\text{max}} - x\right) \\
&\leq  \frac{x + 2n}{2\sqrt{n}-1} -  \frac{2\sqrt{n}x +2n}{2\sqrt{n}-1} + 1 + x \\
&=  \frac{x + 2n}{2\sqrt{n}-1} -  \frac{2\sqrt{n}x +2n}{2\sqrt{n}-1}   + \frac{2\sqrt{n}x-x}{2\sqrt{n}-1} + 1 = 1
\end{align*}
and 
\begin{align*}
s_{j_{\text{max}}} + 2\sqrt{n} - f_3 &=  x - j_{\text{max}}\left(1 - \frac{1}{\sqrt{n}}\right) + 2\sqrt{n} -  \frac{x + 2n}{2\sqrt{n}-1}   \\
&\leq   x - \left(\frac{2\sqrt{n}x +2n}{2\sqrt{n}-1} -1\right) \left(1 - \frac{1}{\sqrt{n}}\right) + 2\sqrt{n} -  \frac{x + 2n}{2\sqrt{n}-1} \\
&=   \frac{2\sqrt{n}x-x}{2\sqrt{n}-1} - \frac{2\sqrt{n}x +2n}{2\sqrt{n}-1} + \frac{2x +2\sqrt{n}}{2\sqrt{n}-1} + 1 - \frac{1}{\sqrt{n}}  +\frac{4n - 2\sqrt{n}}{2\sqrt{n}-1}  -  \frac{x + 2n}{2\sqrt{n}-1} \\
&= 1- \frac{1}{\sqrt{n}}
\end{align*}
Thus, $s_{j_{\text{max}}}, \ldots, s_{j_{\text{min}}}, f_{1}, t_{j_{\text{min}}}, \ldots, t_{j_{\text{max}}}, f_3$
is a sequence of visiting times of the point $x$ with all adjacent visiting times differing in at most 1 time unit, concluding the proof. 
\end{proof}

\section{A schedule for the unidirectonal circle: Proof of Theorem \ref{thm:circle schedule}} \label{sec:circle_schedule}
In this section, we will present a schedule with which a group of agents $a_1, \ldots, a_k$ with maximum speeds $v_1, \ldots, v_k$ can patrol a circle with circumference 
$$ \frac{1}{33\log_e\log_2(k)} \sum_{i=1}^k v_i,$$
followed by a proof of why the proposed schedule behaves as claimed. The construction of our schedule has two steps: we first divide the agents into $\Theta(\log_2 k)$ groups, reducing the speed of some and discarding others so that each group consists of a power of $2$ number of agents that move with the same speed, which is also a power of $2$ times the sum of speeds. This allows us to use a randomized construction, in which the agents from each group are placed equidistantly around the circle with a random offset from some fixed "beginning" of the circle, and move around it with the same speed. We show that with this patrol schedule, most points are visited as frequently as required by our theorem. Then as a second phase, we cut out the bad points -- that is, the ones that are not visited as frequently as necessary. We move the patrol schedule to a smaller circle, intuitively only consisting of the good bits. Agents move as if they were on the larger circle, but whenever moving though a cut-out segment, they just stand still instead. 

Below, we denote by $C_L$ a circle of circumference $L$, formally interpreted as the quotient $\mathbb{R}/L\cdot \mathbb{Z}$.

\begin{proof}[Proof of Theorem \ref{thm:circle schedule}]
We are given a group $a_1, \ldots, a_k$ of agents with maximum speeds $v_1 \ldots, v_k$ and define $V :=\sum_{i=1}^k v_i$. We start by considering the following schedule on $C_1$:
\medskip

\noindent \textbf{Circle Schedule:}
\begin{enumerate}[(1)]
\item Round speeds down to the next power of 2 and omit too slow agents: \\
For all $i \in [k]$ let $j_i \in \mathbb{N}$ be the non-negative integer such that $V\cdot 2^{-j_i} \leq v_i < V\cdot2^{-j_i +1}$. We define $v'_i := V \cdot 2^{-j_i}$ and $I := \left\lbrace i \in [k] : v'_i \geq \frac{V}{4\cdot k} \right\rbrace$.
\item Group remaining agents according to their speed:\\
For all $i \in \{0, \ldots, \lceil \log_2 (k) \rceil + 2 \}$ define $G_i := \left\lbrace j \in I : v'_j =  V\cdot2^{-i} \right\rbrace$.
\item Reduce number of agents in each group to a power of 2:\\
For all $i \in  \{0, \ldots, \lceil \log_2 (k) \rceil + 2 \}$ let $h_i \in \mathbb{N}$ be the positive integer such that $2^{h_i} \leq |G_i| < 2^{h_i+1}$ and let $G'_i \subseteq G_i$ be an arbitrary subset of $G_i$ of size $2^{h_i}$. We denote by $m'_i = |G'_i|\cdot v'_{a} =  V\cdot2^{h_i -i}$ the \emph{mass of the group $G'_i$}, where $a \in G'_i$ (that is, each agent in $G'_i$ has maximum speed $v'_{a}$ after the rounding down in step 1).
\item Omit groups with too small mass:\\
Let $J := \left\lbrace i \in  \{0, \ldots, \lceil \log_2 (k) \rceil + 2  \}: m'_i \geq \frac{V}{16 (\lceil \log_2(k) \rceil+3)} \right\rbrace$.
\item Patrol schedule:
For each $j \in J$, pick an independent uniform random number $r_j$ in the interval $[0,\frac{1}{|G'_j|})$. At time 0, we place the $|G'_j|$ agents from the group $G'_j$ at positions 
$$ r_j, r_j +\frac{1}{|G'_j|}, \ldots, r_j + \frac{|G'_j| -1}{|G'_j|},$$
i.e. we place all agents from the same group equidistantly from each other along the circle with a random offset from the origin. Then the agents $G'_j$  walk along the circle with speed $v'_j$ at all times.
\end{enumerate}

We observe that the two rounding steps and the two omitting steps each reduced the total ``available'' speed by a factor of at most 2, therefore we have that $\sum_{j\in J} m'_j \geq V/16$. Next, define 
$$T = \frac{1}{\min_{j\in J} m'_j} \leq 16(\lceil \log_2(k)\rceil + 3) / V$$
and note that after $T$ time units the distribution of agents along the circle repeats itself, i.e. if an agent with speed $s$ is located at position $x$ at time $t$, then another (or the same) agent of speed $s$ is located at position $x$ at time $t + T$. Additionally, we define $m := \lfloor T_2V/16\log_e\log_2(k)\rfloor$ and $\tau := T_2/m$, where $T_2$ be the largest integer multiple of $T$ which is smaller than $16(\lceil \log_2(k)\rceil + 3) / V$.

We say that a point $x \in C_1$ is \emph{bad} if there exists an $i \in \{0, \ldots, m-1\}$ such that the point $x$ does not get visited in the time interval $[i\tau, (i+1)\tau]$ by any agent. Furthermore, we denote by $B$ the (random) set of bad points and observe that any ``good'' point $x \in C_1 \setminus B$ gets visited at least every $2\tau$ time units. Our goal is to find a suitable upper bound on the expected size of this set $B$ of bad points. We therefore start by calculating the probability that a given point $x \in C_1$ is bad. 

By the union bound we have that $ \Pr[x \in B] \leq  m\cdot \Pr[x \text{ does not get visited in } [0,\tau]]$. As we choose a uniform random offset for each group $G'_j$, the group $G'_j$ covers a uniformly random $\tau\cdot m'_j$ fraction of $C_1$ in the first $\tau$ time units. Thus, we have that  
\begin{align*}
 \Pr[x \in B] &\leq  m\cdot \Pr[x \text{ does not get visited in } [0,\tau]] \\
 &= m\cdot \prod_{j \in J}(1- \tau \cdot m'_j) \\
 &\leq m \cdot e^{-\tau \sum_{j\in J} m'_j} \\
 &\leq m \cdot e^{-\tau \cdot V/16} \\
 &= m\cdot e^{-T_2V/16m}\\
 &\leq \frac{\left\lceil \log_2(k) \right\rceil + 3}{\log_e\log_2(k)} \cdot \frac{1}{\log_2(k)} \leq \frac{2}{\log_e \log_2(k)}.
\end{align*}

Next, we use this to prove the following claim:
\begin{claim}\label{claim:circle 1}
Let $\epsilon >0$ be a fixed constant, and consider the schedule on $C_1$ as above. With probability $1-o(1)$, the set of bad points cover at most an $\epsilon$-fraction of the circle.
\end{claim} 
\begin{proof}
Let $X := \mu(B)$, where $\mu(B)$ is the Lebesgue-measure of the set of bad points $B$. We claim that $\mathbb{E}[X] \leq 2/(\log_e \log_2(k))$. Indeed, by Fubini's Theorem it holds that
\begin{align*}
\mathbb{E}[X] &= \mathbb{E}\left[\int_{0}^{1} \mathds{1}_{\{x \in B\}} \; dx \right] \\
&= \int_{0}^{1} \mathbb{E}\left[\mathds{1}_{\{x \in B\}} \right] dx \\
&= 1\cdot \Pr[x \in B] \leq \frac{2}{\log_e \log_2(k)}
\end{align*}
Hence, by Markov's inequality, 
$$\Pr[X \geq \epsilon] \leq \frac{2}{\epsilon \log_e \log_2(k)} = o(1). $$
Therefore, we have, with probability $1-o(1)$, that a $1-\epsilon$ fraction of the points on the circle $C_1$ get visited at least every $2\tau$ 
 time units, as desired.
\end{proof}

In summary, for any $\epsilon>0$ and  any $k\geq k_0(\epsilon)$, there exists a choice of random offsets in the above strategy such that all but an at most $\epsilon$-fraction of ``bad'' points on $C_1$ is visited every $2\tau$ time steps.

Next, we show that any such patrol schedule can be transformed into a schedule $a'_1(t), \dots, a'_k(t)$ with idle time $2\tau$ on a circle with circumference $1-\mu(B)$. Intuitively, this done by ``cutting out'' intervals of bad points from $C_1$ and gluing together the remaining segments to form a circle of length $1-\mu(B)$. The patrol schedule on $C_{1-\mu(B)}$ mimics the schedule on $C_1$, except that whenever an agent would move along a bad interval on $C_1$, it should instead stand still at the corresponding seam in $C_{1-\mu(B)}$.

Let us formalize the process of cutting out bad points. For any $t_1<t_2$, let $a_i([t_1, t_2])$ denote the closed interval of points on $C_1$ visited by agent $i$ during the time interval $[t_1, t_2]$. We can write the set of good points as
$$ C_1 \setminus B = \bigcap_{j = 0}^{m-1} \bigcup_{i =1}^{k} a_i([j\tau, (j+1)\tau]).$$
Thus, by iterative use of the distributive law for sets, it follows that the set $C_1\setminus B$ of good points can be written as a finite union of closed intervals. Hence, we have a partitioning of $C_1$ into, alternately, good closed intervals and open bad intervals.

Given such a partitioning, we can construct a piecewise linear map $\varphi:C_1\rightarrow C_{1-\mu(B)}$ such that $\varphi(x)$ is constant on any bad interval, and linearly increasing with slope $1$ on any good interval, and construct a patrol schedule $a'_1(t), \dots, a'_k(t)$ on $C_{1-\mu(B)}$ by letting $a'_i(t) := \varphi(a_i(t))$.

As $\varphi$ is non-decreasing with slope at most $1$, it is clear that this is a feasible patrol schedule. Moreover, for any $y\in C_{1-\mu(B)}$, the pre-image $\varphi^{-1}(\{y\})$ is either one good point on $C_1$, or the closure of a bad interval. In either case there exists a good point $x\in C_1\setminus B$ such that $\varphi(x)=y$, and as the patrol schedule on $C_{1-\mu(B)}$ visits $\varphi(x)$ at all times when $x$ gets visited on $C_1$, it follows that the patrol schedule on $C_{1-\mu(B)}$ has idle time $2\tau$, as desired.

Rescaling the patrol schedule $a'_1(t), \ldots, a'_k(t)$, $t \in [0, \infty)$ then gives us the desired patrol schedule with idle time $1$ on a circle $C_L$ of length 
$$L:=(1-\mu(B))/(2\tau) \geq \frac{(1-\epsilon)V}{(32+o(1)) \log_e\log_2(k)} \geq \frac{1}{33\log_e\log_2(k)} \sum_{i=1}^k v_i, $$
concluding the proof.
\end{proof}

\bibliography{refs}

\end{document}